\documentclass[twocolumn]{autart}    
\usepackage{graphicx}          
\usepackage{color}
\usepackage{cite}
\usepackage{amsmath,amssymb,amsfonts}
\usepackage{algorithmic}
\usepackage{graphicx}
\usepackage{textcomp}
\usepackage{amssymb}
\newtheorem{assumption}{Assumption}
\newtheorem{lemma}{Lemma}

\newtheorem{remark}{Remark}
\newtheorem{definition}{Definition}
\newtheorem{proof}{Proof}
\def\dref#1{(\ref{#1})}
\begin{document}

\begin{frontmatter}
\title{ Robust Event-Based Control: Bridge Time-Domain Triggering and Frequency-Domain Uncertainties\thanksref{footnoteinfo}} .

\thanks[footnoteinfo]{This work was supported by Beijing Natural Science Foundation under grant JQ20025 and the National Natural Science Foundation of China under grants 61973006 and T2121002. Corresponding Author: Zhongkui Li.}

\author[China1]{Shiqi Zhang}\ead{zsqpkuedu@pku.edu.cn},    
\author[China1]{Zhongkui Li}\ead{zhongkli@pku.edu.cn}               

\address[China1]{College of Engineering, Peking University, Beijing, China, 100871}  

\begin{abstract} 
This paper considers the robustness of event-triggered control of general linear systems against additive or multiplicative frequency-domain uncertainties.
It is revealed that in static or dynamic event triggering mechanisms, the sampling errors are images of affine operators acting on the sampled outputs. Though not belonging to $\mathcal {RH}_{\infty}$, these operators are finite-gain $\mathcal L_2$ stable with operator-norm depending on  the triggering conditions and the norm bound of the uncertainties.  This characterization is further extended to the general integral quadratic constraint (IQC)-based triggering mechanism. As long as the triggering condition characterizes an $\mathcal L_2$-to-$\mathcal L_2$ mapping relationship (in other words, small-gain-type constraints) between the sampled outputs and the sampling errors, the robust event-triggered controller design problem can be transformed into the standard $H_{\infty}$ synthesis problem of a linear system having the same order as the controlled plant. Algorithms are provided to construct the robust controllers for the static, dynamic and  IQC-based event triggering cases. 
\end{abstract}
\end{frontmatter}

\section{Introduction}

In  control practice such as that of robots, unmanned vehicles and electronic elements, more and more implementations rely on the digital platforms, in which sampling control plays a major role among various control methods. Compared with the traditional control techniques  that  are designed to measure, compute and actuate  continuously, digital controllers sample the output signals discretely, rebuild the estimation via interpolation and then develop controllers  based on the estimation. 
With the development of digitalization, various kinds  of digital algorithms corresponding to different sampling techniques are proposed and the widely used one is periodic sampling control \cite{chenandfrancis1991}. 
Beyond the periodic-sampling control technique, researchers have also developed a variety of aperiodic-sampling control strategies \cite{Laurentiu2017}, which allow to understand the behavior of networked control systems with sampling jitters, packet dropouts or fluctuations due to the interaction between control algorithms and real-time scheduling. First introduced in \cite{Astrom1999}, event-based control (also called Lebesgue sampling control), as a special form  of the aperiodic-sampling control, has been attracting more and more attentions in the control community over the past two decades \cite{tabuada2007,Henningsson2008,Lunze2010,Gawthrop2009,Lemmon2011,xiao2021}. 

The essence of  event-based control is to introduce a triggering function whose value is used as an indicator to decide whether or not  the system is going to do a new sampling.  
Typically, the triggering function is positively correlated to the sampling error and the triggering mechanism is designed such that the sampling will be executed only when the sampling error accumulates and exceeds a certain `triggering level'.
Prominently reducing the frequency of signal sampling and controller updating, event-triggered control can attain the control objectives as well as guarantee the control performance with the triggering mechanism properly designed and thus provides a new control design option  that can help save control energy and resources.

In \cite{Lunze2010}, the performance of the event-triggered control method is shown to be able to  arbitrarily close to continuous control method. 
\cite{Antoine2015} develops  dynamic event-triggered mechanism in order to further reduce the  load of the sampling and to bring more degree of freedom to the controller design.  
In \cite{Donkers2012}, a dynamic output feedback controller as well as a decentralized event-triggered mechanism  is provided to guarantee internal stability and $\mathcal L_{\infty}$ performance.    
\cite{Dolk2017} considers the $\mathcal L_p$ performance for decentralized event-triggered control algorithms.
In \cite{xiao2021},  the authors tackle the stability as well as the $\mathcal L_2$ performance of  networked linear systems  with asynchronous continuous-time or discrete-time event triggering.  
The event-triggered mechanism has also been applied and extended to various other  research domains, such as nonlinear systems \cite{Postoyan2015}, multi-agent  systems \cite{ding2018}, \cite{Nowzari2019}, multi-UAV systems \cite{Ristevski2021}, and so on. 

In practice, uncertainties are ubiquitous and almost everywhere. For a control algorithm to be usable,  robustness  against uncertainties  is  fundamental and thus should be ensured in high priority before other control objectives are considered. Many efforts have been made to tackle the robustness problem of the event-triggered control algorithms. 
For example, \cite{Tripathy2017} and \cite{Kishida2019} consider the robust event-triggered stabilization problem for discrete-time uncertain systems while \cite{Seuret2019} and \cite{Insight} study the continuous-time case. \cite{staticanddynamic} develops the event-triggered state-feedback controllers for a class of nonlinear systems subject to Lipschitz constraints.  \cite{xing2019} investigates the internal stability  of a class of uncertain nonlinear systems under a new event-triggered mechanism via Lyapunov analysis method. In \cite{liu2021}, the authors address the event-triggered control problem for linear systems with uncertain parameters using the hybrid-system approach.  
It should be stressed that these aforementioned works either only consider some special types of system models or can only handle some special forms of uncertainties  such as parametric uncertainties in the time domain. A larger class of uncertainties, namely, frequency-domain uncertainties, is more pervasive in practice \cite{linearrobustcontrol}. Therefore, to investigate the robustness of event-triggered control with respect to frequency-domain uncertainties is evidently an  interesting and important topic, which however remains unsolved.

This paper intends to develop a systematic method for addressing the robust event-triggered output-feedback control problem of general linear systems perturbed by frequency-domain uncertainties, 
which is definitely a quite challenging task. The main difficulty is due to the fact that the uncertainties described as transfer functions in the frequency domain are not evidently compatible with the event-triggered sampling mechanism in the time domain. The gap needs to bridge them in order to design a time-domain event-based controller to stabilize a linear system in the presence of frequency-domain uncertainties. 
Moreover, it is harder to exclude the Zeno behavior under frequency-domain uncertainties, since Zeno behavior is also a continuous sampling phenomenon described in the time domain.  

One contribution of this paper is that under the static and dynamic quadratic-form event triggering mechanisms, the sampling errors are shown to be images of affine operators acting on the sampled outputs. These operators, though not belonging to $\mathcal {RH}_{\infty}$, the set of all real rational stable transfer functions, are finite-gain $\mathcal L_2$ stable, whose operator-norm depends on the triggering conditions and the norm bound of the uncertainties.  
This operator-theoretic approach paves the way to address the robust event-triggered control problem under frequency-domain uncertainties by utilizing classical robust control tools such as the small gain theorem and the integral quadratic constraint (IQC) theory.
 The robust event-triggered output-feedback stabilizing controller is then designed by solving the standard $H_{\infty}$ synthesis problem of a linear system having the same order as the controlled plant. It should be noted that even though only additive and multiplicative uncertainties are considered in this paper, the proposed operator-theoretic approach could easily handle other type of frequency-domain uncertainties and even nonlinear uncertainties with Lipschitz constraints, the latter of which are also finite-gain $\mathcal L_2$-stable \cite{linearrobustcontrol}.
 
The event triggering mechanisms are further extended to the IQC-based event triggering ones. 
Including the static and dynamic event-triggered mechanisms as special cases, the IQC-based event-triggered mechanism are more general, giving more freedom for control performance optimization,  since it allows the parameters in the triggering condition to be dynamical transfer functions which can also be coupled with each other. As long as the designed triggering condition implies an IQC that corresponds  to a `small-gain' quadratic functional, the robust event-triggered control law can be designed in the same framework.

 It is worthy mentioning that the robustness of the event-triggered consensus network against frequency-domain uncertainties is studied in \cite{zhang&li2021} from a similar operator-theoretic viewpoint. Nonetheless, only single-integrator systems and state feedback control are considered in \cite{zhang&li2021}.  In contrast, this paper addresses general high-order linear systems and the output feedback case. More complex system dynamics and the absence of state measurements impose a considerable difficulty in designing feedback gain matrices of the output-feedback controllers; the issue does not exist with single-integrator systems in \cite{zhang&li2021}, where only scalar state feedback gains need to be designed. Moreover, two event-triggered samplers exist between the plant and the controller in this paper, which brings additional difficulties.

The remaining part of this paper is organized as follows.  Some necessary preliminaries on operator theory and robust control are provided in Section \ref{s2}. The problem formulation is presented in Section \ref{s3}. Static event-triggered robust stabilizing control for  linear systems subject to additive and multiplicative uncertainties is considered in Section \ref{s4}. Dynamic event-triggered protocols are further considered in Section \ref{s6}. Section \ref{IQCsection}  extends the robust control problem to the general IQC-based event-triggered mechanism. Section \ref{s7} presents the simulation results and Section \ref{s8} finally concludes this paper.

\section{Mathematical Preliminaries}\label{s2}
\subsection{Operator Theory and Banach  Space }\label{A}
 
	
	
	
\begin{definition}
	The $\mathcal L_2$ norm  of a signal $f$  is defined as $$\|f\|_2=\sqrt{\int_0^{\infty}f^*(t)f(t)dt}=\sqrt{\frac{1}{2\pi}\int_{-\infty}^{\infty}f^*(j\omega)f(j\omega)d\omega},$$ where $\mathcal L_2$ denotes the Banach space with $\mathcal L_2$ norm well defined. 
\end{definition}
\begin{definition}
	\cite{Desoer1975feedback} Letting $\phi(\cdot):\mathcal L_2\mapsto\mathcal L_2$ be an operator  such that for $\forall x\in \mathcal L_2$, $\|\phi(x)\|_2\leq\gamma\|x\|_2+\beta$, where $\gamma>0$ and $\beta>0$ are positive constants, this operator is called a finite-gain $\mathcal L_2$ stable operator with operator norm $\|\phi\|_{\infty}\leq\gamma$.
\end{definition}
\begin{lemma}\label{l2}
	\cite{Desoer1975feedback} Let $\phi(\cdot):\mathcal{L}_2\mapsto \mathcal{L}_2$ denote an affine  (finite-gain $\mathcal L_2$ stable)  operator in the time domain and suppose that $y(t)=\phi(x(t))$, where $y(t)$ and $x(t)$ are vectors in the  $\mathcal L_2$ space. Denote by $y(s)$ and $x(s)$ the Laplace transformation of $y(t)$ and $x(t)$. It then follows that $y(s)=\Delta(x(s))$, where $  \Delta(\cdot):\mathcal L_2\mapsto \mathcal L_2$ denotes a affine  (finite-gain $\mathcal L_2$ stable) operator in the frequency domain. 
	\end{lemma}

\subsection{Robust Control Theory}\label{B}

\begin{lemma}\label{smallgain}
	\cite{Desoer1975feedback} (Small Gain Theorem) Supposing that $G(\cdot), \Delta(\cdot):\mathcal L_2\mapsto\mathcal L_2$ are finite-gain $\mathcal L_2$ stable operators with operator norms $\|G\|_{\infty}=\gamma_1$ and $\|\Delta\|_{\infty}=\gamma_2$. Moreover, $v(s)= G(s)w(s)$ and $w(s)=\Delta(v(s))$. The system interconnection composed of $G(\cdot)$ and $\Delta(\cdot)$ is internally stable, if $\gamma_1\gamma_2<1$.
 \end{lemma}

\section{Problem Formulation}\label{s3}
\subsection{Feedback Loop Structure}\label{structure}
In this paper, we focus on the robust stabilization control problem of an uncertain system $P_{\Delta}(s)$. More specifically, the uncertain systems we study can be represented  as a strictly proper linear nominal system $$P_0(s)=\left[\begin{smallmatrix} 
\begin{array}{c|c}
	A  & B \\ \hline 
	C &  0
\end{array}
\end{smallmatrix}\right]$$  subject to various kinds of uncertainties. For example,   we have  $P_{\Delta}(s)=P_0(s)+\Delta(s)$ for a linear system with an additive perturbation and  $P_{\Delta}(s)=(I+\Delta(s))P_0(s)$ for a multiplicative perturbation. To achieve the control goal, a linear feedback controller 

$$ K(s)=\left[\begin{smallmatrix} 
\begin{array}{c|c}
	A_k  & B_k \\ \hline 
	C_k &  D_k
\end{array}
\end{smallmatrix}\right] $$
is to be designed. 

Throughout this paper, the following assumption holds.
\begin{assumption}\label{as1}
	$(A,B)$ is stabilizable and $(A,C)$ is detectible.
\end{assumption}
There are two samplers $S_1$ and $S_2$  in the feedback loop used to sample the output value $y$ of the plant $P_{\Delta}(s)$ and $u$ of the controller $K(s)$, respectively. The structure of the closed-loop system is depicted in Fig. \ref{fig1}.

\subsection{Event-Triggered Mechanism}\label{mechanism}
To save the computation resources and reduce the frequency of the control updating, we adopt the  event-triggered sampling mechanism proposed in \cite{Heelels2012}.
Here we will  explain specifically  how the sampling system operates  in the event-triggering mechanism. 
The sampler $S_1$ detects the real output $y$ of the  plant $P_{\Delta}(s)$, and calculates the error between it and the estimate of that output $\hat y$ it holds. We use $\epsilon_y=\hat y -y$ to denote this sampling  error. 
Similarly, the sampler $S_2$ detects the real value of the output $u$ of the controller $K(s)$, holds an estimate of it and calculates the sampling  error $\epsilon_u=\hat u-u$.
A  triggering function $f=f(\epsilon_y,\epsilon_u, \hat y, \hat u, t)$ needs be designed  to decide when is the next triggering instant, i.e., when the next sampling has to be done. Typically, we set $f\geq0$ to indicate the next triggering instant. 
At the next triggering instant, say $t_k$, the samplers $S_1$ and $S_2$ update the estimate $\hat y$ to $y(t_k)$ and $\hat u$ to $u(t_k)$ and send them to the controller $K$ and the plant $P_{\Delta}$, respectively. During the two triggering instants, both the estimate $\hat y$ of the output of the plant and the estimate $\hat u$ of the output of the controller  remain constant.

\subsection{Robust Control Objective}
Designing a robust feedback controller $K(s)$ as well as the triggering mechanism to stabilize the perturbed system $P_{\Delta}(s)$ is the essential control objective to be handled in this paper. Meanwhile, the closed-loop system must not exhibit Zeno behavior.
The stability we need to achieve is the so-called `internal stability' \cite{zhou1998essentials}, which requires that disturbance injected from any position in the feedback loop shown in Fig. \ref{fig1} will not lead any divergence phenomenon of the  internal signals.

 \begin{figure}
\begin{center}
\includegraphics[height=6cm]{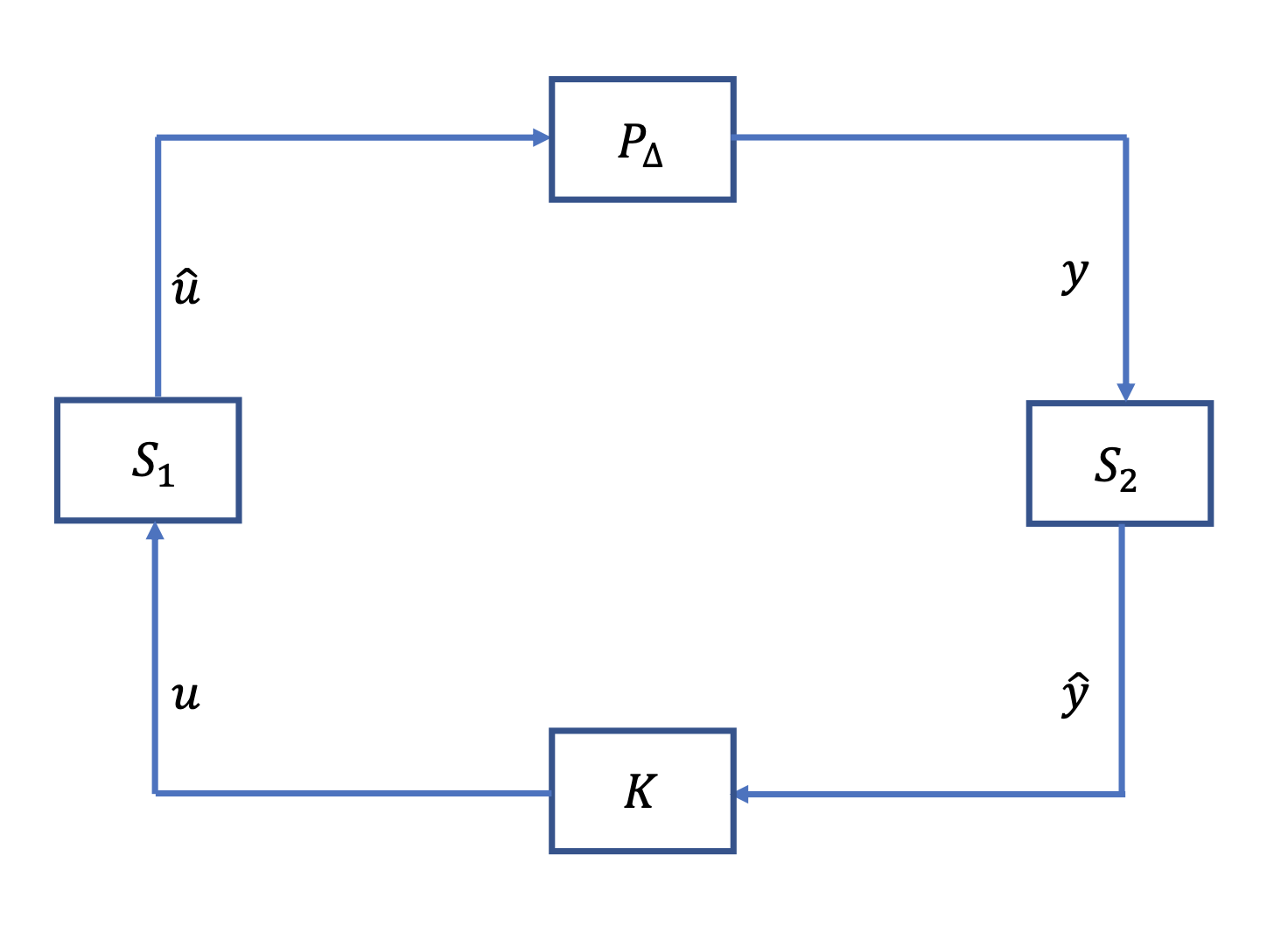}      
\caption{Block diagram representation of the closed-loop system, where $P_{\Delta}(s)$ is the perturbed linear system to be controlled, $K$ is the controller to be developed, $S_1$ and $S_2$ are two event-triggered samplers. }  
\label{fig1}                                 
\end{center}                                 
\end{figure}

\section{Robust Event-Triggered  Stabilizing Control}\label{s4}
\subsection{Linear Systems with Additive Uncertainties}
In this subsection, we aim to deal with the case where $P_{\Delta}(s)$ can be seen as a nominal linear plant $P_0(s)$  perturbed by an additive dynamic uncertainty $\Delta(s)$, i.e., $P_{\Delta}(s)=P_0(s)+\Delta(s)$ with $$\Delta(s)=\left[\begin{smallmatrix} 
\begin{array}{c|c}
	A_{\Delta}  & B_{\Delta} \\ \hline 
	C_{\Delta} &  D_{\Delta}
\end{array}
\end{smallmatrix}\right]\in \mathcal{RH}_{\infty}$$ and $\|\Delta(s)\|_{\infty}\leq\eta$. 

Expressing the above system in the state-space form, we have \begin{equation}\label{eq1}
\begin{aligned}
	\dot x_p &= Ax_p+B\hat u,\\
	y &= Cx_p +d,
\end{aligned}\end{equation}
and \begin{equation}\label{eq2}
\begin{aligned}
	\dot \xi &=A_{\Delta}\xi +B_{\Delta}\hat u,\\
	d  & = C_{\Delta}\xi,
\end{aligned}
	\end{equation} where $x_p, \xi$ denote the internal states of the nominal linear system $P_0(s)$ and the perturbation $\Delta(s)$, respectively, and  $A,B,C,A_{\Delta},B_{\Delta},C_{\Delta}$ are matrices with compatible dimensions.

Also, the controller $K(s)$ can be written in the state-space form as \begin{equation}\label{eq3}
	\begin{aligned}
		\dot x_k&=A_k x_k+B_k\hat y,\\
		u& = C_k x_k+D_k \hat y,
	\end{aligned}
\end{equation}
where $x_k$ denotes the internal state of $K(s)$.
In this section, we adopt the triggering function developed in \cite{Heelels2012} with some modifications:
\begin{equation}\label{triggercondition}
	f=\begin{bmatrix}
	\epsilon_y\\ \epsilon_u
\end{bmatrix}^T\begin{bmatrix}
	\epsilon_y\\ \epsilon_u
\end{bmatrix}-\begin{bmatrix}
	\hat y\\ \hat u 
\end{bmatrix}^T\begin{bmatrix}
	\Omega_1 & \\ & \Omega_2
\end{bmatrix}\begin{bmatrix}
	\hat y\\ \hat u 
\end{bmatrix}-\mu e^{-\nu t},
\end{equation}
where $\Omega_1>0,\Omega_2>0,\mu>0,\nu>0$ are the event-triggering parameters to be determined.

It should be noted that the sampling error $\begin{bmatrix}
\epsilon_y^T & \epsilon_u^T 
\end{bmatrix}^T$ is  updated to zero instantly after each triggering instant, then increases during the two triggering instants and will be reset to zero again at the next triggering instant. The triggering condition ensures that the upper bound of the sampling error is closely related to $\hat y$ and $\hat u$. Actually, this relationship can be characterized as an affine mapping (operator). The next theorem characterizes  this relationship.

\begin{thm}\label{thm1}
Define $\epsilon=[\epsilon_y^T,\epsilon_u^T]^T$ and $v=[\hat y^T, \hat u^T]^T$ and $\gamma=\sqrt{\max\{\rho(\Omega_1),\rho(\Omega_2+\eta^2I)\}}$.  It then follows that under the triggering mechanism described in Section \ref{mechanism}, with the triggering function designed above, we have 
$$\begin{bmatrix}
	\epsilon(s)\\ d(s)
\end{bmatrix}=\Gamma(s)v(s),$$ where $\Gamma(\cdot)$ is an affine  operator whose expression will be given in the proof. Moreover,  $\Gamma(\cdot)$ is finite-gain $\mathcal L_2$-stable with operator gain $\|\Gamma\|_{\infty}\leq \gamma$.
\end{thm}

\begin{proof}
	When $t\in[t_k,t_{k+1})$, $$\begin{aligned}
		\epsilon_y& =\hat y -y\\
		&=-Ce^{A(t-t_k)}x(t_k)-\int_{t_k}^t
Ce^{A(t-\tau)}B\hat u(\tau) d \tau-d(t)\\
 & \quad +\hat y.\end{aligned}$$
 Noticing that $d(s)=\Delta(s)\hat u(s)$, thereby by Lemma \ref{l2} we have $d(t)=\delta(\hat u(t))$, where $\delta(\cdot)$ is a linear operator in the time domain. Denoting $-\int_{t_k}^t
Ce^{A(t-\tau)}B\hat u(\tau) d \tau$ as $\phi(\hat u(t))$ with $\phi(\cdot)$ a linear operator. Then it follows that 
 $$\begin{aligned}
 \epsilon_y & =-Ce^{A(t-t_k)}x(t_k)+\phi(\hat u(t))-\delta(\hat u(t))+\hat y(t) \\
 & =\phi'(\hat u(t)) + \phi''(\hat y(t)) ,
 	\end{aligned}$$
 	where $\phi'(\cdot)$ and   $\phi''(\cdot)$ are affine operators in the time domain. 
 Therefore, in light of Lemma \ref{l2}, we have that $\epsilon_y(s)=\Phi_2(\hat u(s))+\Phi_1(\hat y(s))=\Phi(v(s))$ where $\Phi_2(\cdot)$, $\Phi_1(\cdot)$ and $\Phi(\cdot)$ are affine operators in the frequency domain. 
 Similarly, $$\begin{aligned}
 	\epsilon_u &= \hat u-u\\
 	& = \hat u-D_k \hat y-\int_{t_k}^tC_k e^{A_k(t-\tau)}B_k\hat y(\tau)d \tau\\
 	&\quad-C_ke^{A_k(t-t_k)}x_k(t_k)\\
 	& = \psi''(\hat u(t))
+\psi'(\hat y(t)), \end{aligned}$$
 where $\psi'(\cdot)$ and  $\psi''(\cdot)$ are affine operators in the time domain.
Therefore, $\epsilon_u(s)=\Psi_1(\hat y(s))+\Psi_2(\hat u(s))=\Psi(v(s))$ where $\Psi_1(\cdot)$, $\Psi_2(\cdot)$ and $\Psi(\cdot)$ are affine operators in the frequency domain.
 Since $\Delta(s)\in \mathcal{RH}_{\infty}$,  it is definitely an affine operator in the frequency domain. Thus, $d(s)=\Delta(u(s))=\Delta'(v(s))$, where $\Delta'(\cdot)$ is also an affine operator in the frequency domain. 
 
Based on the above discussions, we have 
$$\begin{bmatrix}
	\epsilon(s) \\ d(s)
\end{bmatrix}=\begin{bmatrix}
	\Phi(\cdot) \\ \Psi(\cdot)  \\  \Delta'(\cdot)
\end{bmatrix}v(s)=\Gamma(v(s)),$$
where $\Gamma(\cdot)=\begin{bmatrix}
	\Phi(\cdot) \\ \Psi(\cdot)  \\  \Delta'(\cdot)
\end{bmatrix} $ is an affine operator. 
Next, we prove that the affine-operator $\Gamma(\cdot)$ is finite-gain $\mathcal{L}_2$ stable and calculate its $\mathcal L_2$ norm bound. 
Since $d(s)=\Delta(s)\hat u(s)$ and $\|\Delta\|_{\infty}\leq \eta$, we have 
\begin{equation}\label{eq4}
	\int_{-\infty}^{\infty}d^*(j\omega)d(j\omega)d\omega\leq\eta^2\int_{-\infty}^{\infty}\hat u^*(j\omega)\hat u(j\omega)d\omega.
\end{equation}
From the triggering condition shown in \dref{triggercondition}, it is clear that $$
\begin{aligned}
	\int_0^{\infty} f(t)dt = & \int_0^{\infty}
	\epsilon^T
	\epsilon dt\\
 & -\int_0^{\infty}\begin{bmatrix}
	\hat y\\ \hat u 
\end{bmatrix}^T\begin{bmatrix}
	\Omega_1 & \\ & \Omega_2
\end{bmatrix}\begin{bmatrix}
	\hat y\\ \hat u 
\end{bmatrix}dt-\frac{\mu}{\nu} \leq 0.
\end{aligned}
$$ 
In light of the Parseval identity \cite{Desoer1975feedback}, it is easy to see that
\begin{equation}\label{eq5}
	\begin{aligned}
	&\frac{1}{2\pi}\int_{-\infty}^{\infty}\epsilon^*(j\omega)\epsilon(j\omega) d\omega\\ \leq & \frac{1}{2\pi}\int_{-\infty}^{\infty}\begin{bmatrix}
	\hat y(j\omega)\\ \hat u(j\omega) 
\end{bmatrix}^*\begin{bmatrix}
	\Omega_1 & \\ & \Omega_2
\end{bmatrix}\begin{bmatrix}
	\hat y(j\omega)\\ \hat u (j\omega)
\end{bmatrix}d\omega+\frac{\mu}{\nu}.
\end{aligned}
\end{equation}
Adding \dref{eq4} and \dref{eq5}, we can  get that 
\begin{equation}\label{eq6}
	\begin{aligned}
		&\quad\int_{-\infty}^{\infty}\begin{bmatrix}
			\epsilon(j\omega)\\ d(j\omega)
		\end{bmatrix}^*\begin{bmatrix}
			\epsilon(j\omega)\\ d(j\omega)
		\end{bmatrix} d\omega
		\\
		&\leq  \int_{-\infty}^{\infty}v^*(j\omega)\begin{bmatrix}
			\Omega_1 &\\
			& \Omega_2+\eta^2I
		\end{bmatrix}v(j\omega)d\omega+\frac{2\pi\mu}{\nu} 
		\\&\leq \gamma^2\int_{-\infty}^{\infty}v^*(j\omega)v(j\omega)d\omega+\frac{2\pi\mu}{\nu}.
	\end{aligned}
\end{equation}
Based on the definition of finite-gain $\mathcal{L}_2$ stability and the ${H}_{\infty}$ norm theory, obviously $\Gamma(\cdot)$ is finite-gain $\mathcal L_2$ stable  with  operator norm $\|\Gamma\|_{\infty}\leq\gamma$. $\hfill\square$
\end{proof}
\begin{remark}
	This theorem gives an explicit relationship between the sampling error $\epsilon$ and the sampled  output value   $v$. The sampling error is, in fact, an image of an affine mapping  acting on the sampled output.  Notice that this operator can be denoted as $\begin{bmatrix}
	\Phi(\cdot) \\  \Psi(\cdot) 
	\end{bmatrix}$ and is in general not a transfer function in $\mathcal{RH}_{\infty}$, but it is finite-gain $\mathcal L_2$ stable, i.e., it maps a causal $\mathcal L_2$ signal to another causal $\mathcal L_2$ signal under the event-triggered mechanism described above.
\end{remark}
\begin{figure}
\begin{center}
\includegraphics[height=9cm,width=9cm]{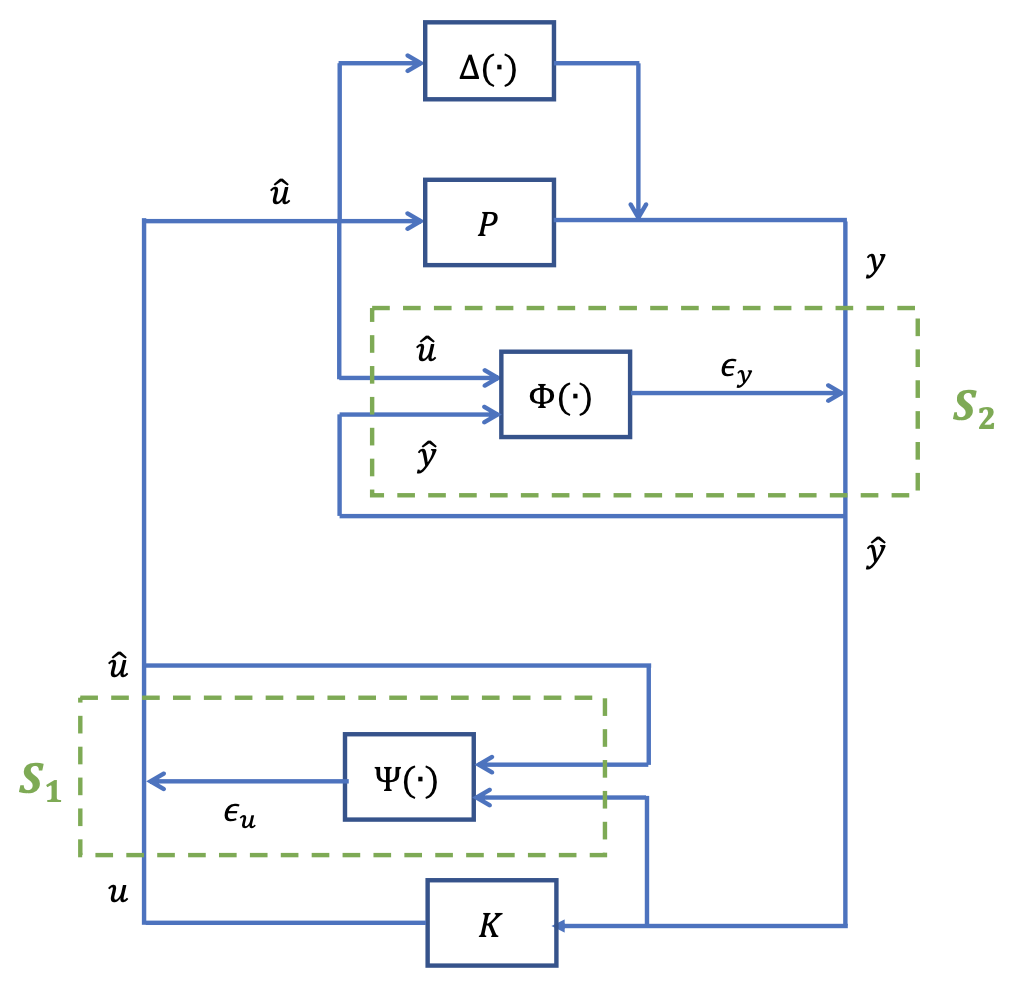}      
\caption{ Block diagram representation of the closed-loop system. The controlled plant is perturbed by an additive dynamic uncertainty and the effect of the two event-triggering sampler $S_2$ and $S_1$ is equivalent to introducing two uncertain affine operators $\Psi(\cdot)$ and $\Phi(\cdot)$ into the feedback loop that act on the sampled output $\hat y$ and $\hat u$ to generate the sampling errors $\epsilon_y$ and $\epsilon_u$, respectively. }  
\label{fig2}                                 
\end{center}                                 
\end{figure}
\begin{remark}
	According  to Theorem \ref{thm1}, the system interconnection shown in Fig. \ref{fig1} is actually equivalent to the system interconnection shown in Fig. \ref{fig2}. The two event-triggered samplers $S_1$ and $S_2$   in Fig. \ref{fig1} are equivalent to the  two  block structure units contained in the two green boxes in  Fig. \ref{fig2}. To be more specific, the effect of the sampler $S_2$  can be seen as acting on the sampled output $\hat y$ and $\hat u$ to generate the sampling error $\epsilon_y$, while that of the other sampler $S_1$ is to generate $\epsilon_u$. 
	It should be noted that the explicit form of the operators $\Psi(\cdot)$ and $\Phi(\cdot)$ can be very messy and it is in general impossible to express them in  a closed form. However, this does not prevent from using the properties of the operators such as $\mathcal L_2$ stability, since the  key quantity we are  interested in is the size or norm of the operators, which will be deliberated in the sequal.
\end{remark}
Before we proceed to determine the controller parameters and the parameters in the event-triggered mechanism, we first do some transformations to the feedback loop to get some more insightful results. 
More specifically, we `pull out' the uncertain operator $\Gamma(\cdot)$. Notice that $\hat u= u+\epsilon_u$ and $\hat y= y+\epsilon_y$, we  can then derive that 
\begin{equation}\label{close}
	\begin{aligned}
		\begin{bmatrix}
			\dot x_p\\ \dot x_k  
		\end{bmatrix}&=A_{cl}\begin{bmatrix}
			x_p\\ x_k
		\end{bmatrix}
		 +B_{cl}\begin{bmatrix}
			\epsilon_y\\ \epsilon_u\\ d 
		\end{bmatrix},\\
		\begin{bmatrix}
			\hat y\\ \hat u 
		\end{bmatrix}&=C_{cl}\begin{bmatrix}
			x_p\\ x_k
		\end{bmatrix}+D_{cl}\begin{bmatrix}
			\epsilon_y\\ \epsilon_u\\ d 
		\end{bmatrix},
	\end{aligned}
\end{equation}
where $A_{cl}=\begin{bmatrix}
			A+BD_kC & BC_k\\
			B_k C & A_k
		\end{bmatrix}$, $B_{cl}=\begin{bmatrix}
			BD_k & B & BD_k\\
			B_k & 0 & B_k 
		\end{bmatrix}$, $C_{cl}=\begin{bmatrix}
			C & 0 \\ D_kC  & C_k
		\end{bmatrix}$, $D_{cl}=\begin{bmatrix}
			I & 0 & I\\
			D_k & I & D_k
		\end{bmatrix}$.

Denote the system in \dref{close} as $v(s)= M(s)\begin{bmatrix}
	\epsilon(s)\\ d(s)
\end{bmatrix}$, where $M(s)=\left[\begin{smallmatrix} 
\begin{array}{c|c}
	A_{cl}  & B_{cl} \\ \hline 
	C_{cl} &  D_{cl}
\end{array}
\end{smallmatrix}\right]$. Combining with the fact that $\begin{bmatrix}
	\epsilon(s)\\ d(s)
\end{bmatrix}=\Gamma(v(s))$,  we then get a classical feedback loop interconnected by a linear system $M(s)$ and a norm-bounded $\mathcal L_2$ stable but unknown operator $\Gamma(\cdot)$.

Now we turn back to the major task of  designing a stabilizing controller $K(s)=\left[\begin{smallmatrix} 
\begin{array}{c|c}
	A_{k}  & B_{k} \\ \hline 
	C_{k} &  D_{k}
\end{array}
\end{smallmatrix}\right].$ In light of the well-known Small Gain Theorem (Lemma \ref{smallgain}), the problem can then be transformed to designing a controller $K(s)$ such that $M(s)\in \mathcal{RH}_{\infty}$ and $\|M(s)\|_{\infty}<\frac{1}{\gamma}$. This problem can then be coped with  by solving a standard $H_{\infty}$ problem. To see this, consider the following linear system:    $$\begin{bmatrix}
	z(s)\\ y(s)
\end{bmatrix}=P_{add}(s)\begin{bmatrix}
	w(s)\\ u(s)
\end{bmatrix},$$
where $P_{add}(s)=\left[\begin{smallmatrix}
     	\begin{array}{c|cc}
	A  & B_1 & B_2 \\ \hline 
	C_1 &  D_{11} & D_{12}\\
	C_2 & D_{21} & D_{22}
\end{array}
     \end{smallmatrix}\right]$, $B_1=\begin{bmatrix}
     	0 & B& 0 
     \end{bmatrix}$, $B_2=B$, $C_1=\begin{bmatrix}
     	C\\ 0
     \end{bmatrix}$, $C_2=C$, $D_{11}=\begin{bmatrix}
     	I & 0 & I \\ 0 & I & 0
     \end{bmatrix}$, $D_{12}=\begin{bmatrix}
     	0\\ I
     \end{bmatrix}$, $D_{21}=\begin{bmatrix}
     	I & 0 & I
     \end{bmatrix}$ and $D_{22}=0$. It is not difficult to  verify that  
$$T_{wz}=\mathcal F_l(P_0(s),K(s))=\left[\begin{smallmatrix} 
\begin{array}{c|c}
	A_{cl}  & B_{cl} \\ \hline 
	C_{cl} &  D_{cl}
\end{array}
\end{smallmatrix}\right]=M(s).$$ Therefore,  designing a controller $K(s)$ such that $M(s)$ is stable and $\|M(s)\|_{\infty}<\frac{1}{\gamma}$ is equivalent to designing such a $K(s)$ such   that $T_{wz}(s)=\mathcal F_l(P_{add}(s),K(s))$ is stable and $\|T_{wz}\|_{\infty}=\|\mathcal F_l(P_{add}(s),K(s))\|_{\infty}<\frac{1}{\gamma}$. This result can be concluded in the following theorem: 
\begin{thm}\label{thm2}
     Assume that $\gamma>0$, an event-triggered controller with sampling mechanism shown in Section \ref{mechanism}, triggering condition shown as  \dref{triggercondition} and    $K(s)=\left[\begin{smallmatrix} 
\begin{array}{c|c}
	A_{k}  & B_{k} \\ \hline 
	C_{k} &  D_{k}
\end{array}
\end{smallmatrix}\right]$ stabilizes any perturbed plant $P_{\Delta}(s)=P_0(s)+\Delta(s)$ with $\|\Delta\|_{\infty}\leq\eta$,  if and only if $K(s)$ solves   a standard dynamic output feedback $\gamma^{-1}$-suboptimal $H_{\infty}$ synthesis problem with $P_{add}(s)$ defined as above and $\gamma$ determined by Theorem \ref{thm1}. 
\end{thm}
 
\begin{remark}
Actually, we can further relax our assumption in the sense that  the uncertainty  $\Delta(s)$  does not even need to be a transfer matrix in $\mathcal {RH}_{\infty}$. All we need is that $\Delta(s)$ is finite-gain $\mathcal L_2$ stable. This observation extends the field of application of the developed  method provided above.
\end{remark}

In light of the above  theorem,  we can then design the  event-triggered controller  stabilizing the perturbed system $P_{\Delta}$. The corresponding  parameters $A_k,B_k,C_k,D_k,\Omega_1>0,\Omega_2>0$, $\mu>0$ and $\nu>0$ in the event-triggered controller  can  be determined  via the following algorithm.

\begin{alg}\label{alg1}
	An algorithm to find a robust event-triggered   protocol for $P_{\Delta}=P_0+\Delta$:
	
	Step 1: Solve the standard  $H_{\infty}$ optimal synthesis problem  in Theorem \ref{thm2} and find the optimal $H_{\infty}$ level $\gamma_{opt}$. If $\gamma_{opt}<\eta^{-1}$, go to the next step. Otherwise, the robust event-triggered controller  may not exist.     
	
	 Step 2: Choose positive numbers $\mu>0$, $\nu>0$. Determine positive-definite matrices $0<\Omega_1<\gamma_{opt}^{-2}I$ and $0<\Omega_2<(\gamma_{opt}^{-2}-\eta^2)I$ and calculate $\gamma=\sqrt{\max\{\rho(\Omega_1),\rho(\Omega_2+\eta^2I)\}}$.
	 
	 Step 3: Solve the standard $\gamma^{-1}$-suboptimal $H_{\infty}$ synthesis problem  in Theorem \ref{thm2} to get the controller  $K(s)=\left[\begin{smallmatrix} 
\begin{array}{c|c}
	A_{k}  & B_{k} \\ \hline 
	C_{k} &  D_{k}
\end{array}
\end{smallmatrix}\right]$.

\end{alg}  

\begin{remark}
	Note that this algorithm eventually transforms the event-triggered robust stabilizing controller synthesis problem to solving a standard $H_{\infty}$ suboptimal control problem with the same order as the original nominal linear plant and that  the parameters in the $H_{\infty}$ suboptimal problem all depend on the already-known parameters of the  nominal plant.  The suboptimal $H_{\infty}$ level to be achieved is determined by the parameters in the triggering mechanism which remains to be developed. 
\end{remark}

The next theorem shows that the closed-loop system does not exhibit Zeno behavior.  
\begin{thm}\label{thm3}
	The system interconnection resulted by  the event-triggering sampling mechanism designed in Section \ref{mechanism} and the control protocol designed in  Algorithm \ref{alg1} does not exhibit Zeno behavior.
\end{thm}    
\begin{proof}
    We will exclude Zeno behavior by contradiction. For simplicity, we assume that $D_{\Delta}=0$ in $\Delta(s)$.
    
    First, since the closed-loop system is internally stable, then we have that $\zeta=\begin{bmatrix}
    	x_p^T  & x_k^T & \xi^T
    \end{bmatrix}^T$ is a vector with a norm bound, say $H$. During each time interval between the two consecutive triggering instants, say, $[t_k,t_{k+1})$, taking the time derivative of the vector $\zeta$, we will get that
    \begin{equation}\label{dote}
    	\begin{aligned}
    		\dot\epsilon=\Hat A \epsilon + \hat B\zeta,
    	\end{aligned}
    \end{equation} 
    where $$\begin{aligned} \hat A &= \begin{bmatrix}
    	-CBD_k-C_{\Delta}B_{\Delta}D_k & -CB-C_{\Delta}B_{\Delta}\\
    	-C_kB_k & 0
    \end{bmatrix},\\
    	 \hat  B &= 
    	\begin{bmatrix}
    	 \hat B_{11} & \hat B_{12} & \hat B_{13}\\
    	 \hat B_{21} & \hat B_{22} & \hat B_{23}
    		    	\end{bmatrix}, \\
    \hat B_{11}&=-CA-CBD_kC-C_{\Delta}B_{\Delta}D_kC,\\
    \hat  B_{12}&= -CBC_k-C_{\Delta}B_{\Delta}C_k,\\
    \hat B_{13}& = -CBD_kC_{\Delta}-C_{\Delta}B_{\Delta}D_kC_{\Delta}-C_{\Delta}A_{\Delta},\\
    \hat B_{21}&= -C_kB_kC,\\
    \hat B_{22}& = -C_kA_k ,\\
     \hat B_{23}& = -C_kB_kC_{\Delta}.\end{aligned}$$ 
  Suppose that there exists Zeno behavior. Then there is a time instant $T$ such that $\lim_{k\rightarrow\infty}t_{k}=T<\infty$. This is equivalent to saying that for a small positive real number $\delta=\frac{\mu\|\hat A\|}{\|\hat B\|^2H^2}e^{-(4\|\hat A\|+\nu)T}$, there exists such an integer $K$ such that for all $k \geq K$, we have $t_k\in (T-\delta,T] $. Then we will consider specially the time interval $[t_K,t_{K+1})$.  
	Notice that at the triggering instant $t_K$, $\epsilon_y$ and $\epsilon_u$ are reset to be zero. Since the next triggering time is the first time that $\left\|\epsilon \right\|^2$ reaches $ v^T
\left[\begin{smallmatrix}
	\Omega_1 & \\ & \Omega_2
\end{smallmatrix}\right]
v+\mu e^{-\nu t},$ 
there must exist some time instant $T_0\in [t_K,t_{K+1})$ when $\left\|
	\epsilon(T_0) 
	\right\|^2=\mu e^{-\nu T_0}.$
On the other hand,
$$ 
\begin{aligned}
\left\|
\epsilon(T_0)\right\|^2& =\left\|\int_{t_K}^{T_0}e^{\hat A(T_0-\tau)}\hat B \zeta(\tau)d \tau\right\|^2\\
&\leq \left\|e^{\hat A(T_0-t_K)}\right\|^2\left\|\int_{t_K}^{T_0}e^{-\hat A(\tau-t_K)}\hat B\zeta(\tau)d\tau\right\|^2\\
& \leq e^{2\|\hat A\|T}\left[\int_{t_K}^{T_0}\left\|e^{-\hat A(\tau-t_K)}\right\|\left\|\hat B\zeta(\tau)\right\|d\tau\right]^2.
\end{aligned}
$$
According to the well-known Cauthy-Swartz inequality, it then follows that 
$$
\begin{aligned}
	&\left\|\epsilon(T_0)
\right\|^2 \\
& \leq
e^{2\|\hat A\|T}\left[\int_{t_K}^{T_0}e^{2\|\hat A\|(\tau-t_K)}d\tau\right]\left[\int_{t_K}^{T_0}\|\hat B\|^2\|\zeta(\tau)\|^2d\tau \right]\\
& \leq e^{2\|\hat A\|T}\frac{1}{2\|\hat A\|}e^{2\|\hat A\|(T_0-t_K)}\|\hat B\|^2H^2(T_0-t_K)\\
& \leq e^{2\|\hat A\|T}\frac{1}{2\|\hat A\|} e^{2\|\hat A\|T}\|\hat B\|^2H^2(T_0-t_K).\\
\end{aligned}
$$
Notice also that $\mu e^{-\nu T_0}\geq\mu e^{-\nu T}$, therefore we have 
$$T_0-t_K\geq\frac{2\mu\|\hat A\|}{\|\hat B\|^2H^2}e^{-(4\|\hat A\|+\nu)T}=2\delta.$$
This implies that 
$$t_{K+1}-t_K> \delta\geq T-t_K,$$
which contradicts our assumption. This completes the proof. $\hfill\square$
\end{proof}

\subsection{Linear Systems with  Multiplicative Uncertainties }\label{s5}
In this subsection, we consider a linear nominal plant  perturbed by a multiplicative uncertainty. Therefore, the perturbed plant has the form  $P_{\Delta}(s)=(I+\Delta(s))P_0(s)$ with $\Delta(s)\in \mathcal{RH}_{\infty}$ and $\|\Delta(s)\|_{\infty}\leq\eta$.  
In light of this, we have 
\begin{equation}
	\begin{aligned}\label{P}
		\dot x_p &= A x_p+B\hat u, \\
		\lambda  & = C x_p,
	\end{aligned}
\end{equation}
\begin{equation}\label{Delta}
	\begin{aligned}
		\dot \xi &= A_{\Delta}\xi +B_{\Delta}\lambda,\\
		d  & =  C_{\Delta}\xi,
	\end{aligned}
\end{equation}
\begin{equation}\label{111}
	y=\lambda+d,
\end{equation}
and
\begin{equation}\label{K}
	\begin{aligned}
		\dot x_k & = A_k x_k+B_k \hat y,\\
		u &  = C_k x_k+D_k\hat y,
	\end{aligned}
\end{equation}
where $\hat y$ and $\hat u$ denote the sampled value of $y$ and $u$, respectively.
The triggering mechanism is designed to be  the same as the  one shown in Section \ref{mechanism} and we still adopt the static triggering function \dref{triggercondition}.

Let $\epsilon_y=\hat y-y$, $\epsilon_u = \hat u-u$, $\tilde v=\begin{bmatrix}
	\lambda^T & \hat y^T & \hat u^T 
\end{bmatrix}^T$ and  $\epsilon = \begin{bmatrix}
	\epsilon_y^T & \epsilon_u^T
\end{bmatrix}^T$. 
Similar to the case where the system is subject to an additive uncertainty, in this case by studying the relationship between the vector $\begin{bmatrix}
	d^T & \epsilon^T  
\end{bmatrix}^T$   and the vector $\tilde v$, we have the following result.
\begin{thm}\label{thm4}
	  $$\begin{bmatrix}
	d(s) \\ \epsilon(s)  
\end{bmatrix}=\hat\Gamma(s)\tilde v(s),$$ where $\hat\Gamma(\cdot)$ is an affine operator which is  finite-gain $\mathcal L_2$ stable. Meanwhile, the operator gain $\|\Gamma\|_{\infty}\leq\hat\gamma$, where $\hat\gamma=\sqrt{\max\{\eta^2,\rho(\Omega_1),\rho(\Omega_2)\}}$. 
\end{thm}
\begin{proof}
	The proof is quite similar to that of  Theorem \ref{thm1} and thus is omitted for brevity. $\hfill\square$

\end{proof}

On the other hand, combining \dref{P}, \dref{Delta}, \dref{111} and \dref{K}, we can get that
\begin{equation}\label{close2}
\begin{aligned}
	\begin{bmatrix}
		\dot x_p\\ \dot x_k
	\end{bmatrix}&=\hat A_{cl}\begin{bmatrix}
		x_p\\ x_k
	\end{bmatrix}+\hat B_{cl}\begin{bmatrix}
		d\\ \epsilon
	\end{bmatrix},\\
	\tilde v&=\hat C_{cl}\begin{bmatrix}
		x_p\\ x_k
	\end{bmatrix}+\hat D_{cl}\begin{bmatrix}
		d\\ \epsilon
	\end{bmatrix},
\end{aligned}	
\end{equation}
where  
$$\begin{aligned}
\hat A_{cl}&=\begin{bmatrix}
	A+BD_kC & BC_k\\
	B_kC & A_k
\end{bmatrix},\\
\hat B_{cl}&=\begin{bmatrix}
	BD_k& BD_k & B\\
	B_k & B_k & 0
\end{bmatrix},\\
\hat C_{cl} &=\begin{bmatrix}
		C & 0 \\ C & 0 \\ D_kC  & C_k
	\end{bmatrix},\\
	\hat D_{cl}&=\begin{bmatrix}
		0&0&0\\
		I&I&0\\
		D_k&D_k& 0
	\end{bmatrix}.\end{aligned} $$
 
 Let $\hat M(s)=\left[\begin{smallmatrix} 
\begin{array}{c|c}
	\hat A_{cl}  & \hat  B_{cl} \\ \hline 
	\hat C_{cl} & \hat D_{cl}
\end{array}
\end{smallmatrix}\right]$, we then have $\tilde v(s)=\hat M(s)\begin{bmatrix}
	d(s)\\ \epsilon(s)
\end{bmatrix}.$
Combining with the fact that $\begin{bmatrix}
	d(s) \\ \epsilon(s)  
\end{bmatrix}=\hat\Gamma(v(s))$, we can then get an interconnecting loop consisting of $M(s)$ and $\hat\Gamma(\cdot)$. By the Small Gain Theorem, it is clear that the event-triggered protocol stabilizes the perturbed linear system, if $\hat M(s)\in \mathcal{RH}_{\infty}$ and $\|\hat M(s)\|_{\infty}<\hat \gamma^{-1}$. This design objective, again, can be transformed to a standard dynamic output  $H_{\infty}$ synthesis problem, as will be shown in the next theorem.

\begin{thm}\label{thm5}
	Given a certain $\hat \gamma>0$, we can then find such an event-triggered  controller with    $$K(s)=\left[\begin{smallmatrix} 
\begin{array}{c|c}
	 A_{k}  &   B_{k} \\ \hline 
	 C_{k} & D_{k}
\end{array}
\end{smallmatrix}\right]$$ that makes  the closed-loop system \dref{P},\dref{Delta},\dref{111},\dref{K} internally stable  by solving a standard dynamic output $\hat \gamma^{-1}$-suboptimal $H_{\infty}$ synthesis problem  with 
	$$\hat P_{mul}=\left[\begin{smallmatrix}
     	\begin{array}{c|cc}
	A  & \hat B_1 & \hat B_2 \\ \hline 
	\hat C_1 & \hat  D_{11} & \hat D_{12}\\
	\hat C_2 & \hat D_{21} & \hat D_{22}
\end{array}
     \end{smallmatrix}\right],$$ 
     where $\hat B_1=\begin{bmatrix}
     	0 & 0& B 
     \end{bmatrix}$, $\hat B_2=B$, $\hat C_1=\begin{bmatrix}
     	C\\ C\\ 0 
     \end{bmatrix}$, $\hat C_2=C$, $\hat D_{11}=\begin{bmatrix}
     	0&0&0\\
     	I&I&0 \\
     	0&0&0
     \end{bmatrix}$,
     $\hat D_{12}=\begin{bmatrix}
     	0\\0\\ I
     \end{bmatrix}$, $\hat D_{21}=\begin{bmatrix}
     	I& I& 0
     \end{bmatrix}$ and $\hat D_{22}=0$.
\end{thm}
\begin{proof}
The proof is similar to that of Theorem \ref{thm2} and thus is omitted for brevity.	 $\hfill\square$

\end{proof}

In light of this theorem, we can find such an event-triggered  controller  with triggering mechanism stated in Section \ref{mechanism}, triggering condition in \dref{triggercondition} and $K(s)=\left[\begin{smallmatrix} 
\begin{array}{c|c}
	 A_{k}  &   B_{k} \\ \hline 
	 C_{k} & D_{k}
\end{array}
\end{smallmatrix}\right]$, whose  parameters $A_k$, $B_k$, $C_k$, $D_k$, $\Omega_1>0$, $\Omega_2>0$, $\mu>0$ and $\nu>0$ can be determined according to the following algorithm.

\begin{alg}\label{alg2}
	An algorithm to find a stabilizing event-triggering controller for $P_{\Delta}=(I+\Delta)P_0$:
	
	Step 1: Solve the standard $H_{\infty}$ optimal synthesis problem  in Theorem \ref{thm5} and find the optimal $H_{\infty}$ level $\kappa_{opt}$. If $\kappa_{opt}^{-1}\geq \eta$, go to the next step. Otherwise, the stabilizing event-triggering protocol may not exist.
	
	Step 2: Choose positive numbers $\mu>0$, $\nu>0$. Determine positive-definite matrices $0<\Omega_1<\kappa_{opt}^{-2}I$  and $0<\Omega_2<\kappa_{opt}^{-2}I$ and calculate $\hat \gamma=\sqrt{\max\{\eta^2,\rho(\Omega_1),\rho(\Omega_2)\}}$.
	
	Step 3: Solve the $\hat \gamma^{-1}$-suboptimal $H_{\infty}$ synthesis problem  in Theorem \ref{thm5} to get the controller  $K(s)=\left[\begin{smallmatrix} 
\begin{array}{c|c}
	A_{k}  & B_{k} \\ \hline 
	C_{k} &  D_{k}
\end{array}
\end{smallmatrix}\right]$.
	
\end{alg}
 \begin{thm}\label{thm6}
 	The system interconnection resulted by the event-triggered sampling mechanism shown in subsection \ref{mechanism} and the controller  designed in Algorithm \ref{alg2} does not exhibit Zeno behavior.
 \end{thm}
\begin{proof}
	The proof is omitted for brevity. $\hfill\square$
\end{proof}
\begin{remark}
	Similar to the case of additive uncertainty, in the multiplicative uncertainty case the robust stabilizing controller synthesis problem is also transformed into a standard $H_{\infty}$ suboptimal control problem with the same order as the nominal linear plant $P_0$. The only difference lies  on the system parameters. Actually, through our operator theoretic approach, more kinds of uncertainties are able to be handled, including the coprime factor uncertainty, and even nonlinear dynamics  with Lipschitz constraints, which can also be represented as  an $\mathcal L_2$ stable operator \cite{linearrobustcontrol}. 
\end{remark}

\section{Extensions to Dynamic Event-Triggered Mechanisms}\label{s6}
In this section, we aim to design an event-triggered  controller robust against frequency-domain uncertainties  based on the dynamic event-triggered mechanism. 
To be more specific, 
inspired by \cite{Antoine2015}, we adopt the  dynamic   triggering function described as follows:
\begin{equation}\label{eq14}
\begin{aligned}
	f&=\epsilon^T\epsilon-v^T\begin{bmatrix}
	\Omega_1 & \\ & \Omega_2
\end{bmatrix}v-\chi,
\end{aligned}
\end{equation}
where $\epsilon$ and $v$ are defined as in Theorem \ref{thm1}, $\chi$ is the internal scalar  state with dynamics:
\begin{equation}\label{eq15}
\begin{aligned}
\dot \chi &=-\beta \chi-\alpha \left( \epsilon^T\epsilon-v^T\begin{bmatrix}
	\Omega_1 & \\ & \Omega_2
\end{bmatrix}v\right)
\end{aligned}
\end{equation}
with $\beta>0$, $\alpha>0$ and $\chi(0)>0$.
The samplers $S_1$ and $S_2$ will sample the output of the plant and the controller whenever $f\geq 0$.

For illustration, in this section we will only consider   linear systems perturbed by  an  additive dynamic uncertainty. The system equation is described as \dref{eq1} and \dref{eq2} and the controller is in the form of \dref{eq3}. 
Interestingly, the introduction of dynamic event-triggering mechanism does not change the conclusion we derive in Theorem \ref{thm1}, as shown in the next theorem.

\begin{thm}\label{thm7}
    Let $\epsilon$ and $v$ be the vectors defined as in Theorem \ref{thm1} and $\gamma=\sqrt{\max\{\rho(\Omega_1),\rho(\Omega_2+\eta^2I)\}}$.
	Under the dynamic event-triggered mechanism described above with the triggering function shown in  \dref{eq14}, we have $$\begin{bmatrix}
			\epsilon(s)\\ d(s)
		\end{bmatrix}=\Gamma(s)v(s),$$ where $\Gamma$ is an affine operator which is finite-gain $\mathcal L_2$ stable. Moreover, $\|\Gamma\|_{\infty}\leq \gamma$.
\end{thm}
\begin{proof}
	Based on the triggering condition, we will always have $f\leq 0$, which implies 
	\begin{equation}\label{eq16}
	 \begin{aligned}
	\epsilon^T\epsilon-v^T\begin{bmatrix}
	\Omega_1 & \\ & \Omega_2
\end{bmatrix}v\leq\chi.
\end{aligned}
	\end{equation}
	Therefore, in light of \dref{eq15}, it then follows that $$\dot \chi\geq-(\beta +\alpha)\chi.$$
	Notice also that $\chi(0)>0$, by the well-known Comparison Lemma \cite{Khalil}, we have $\chi(t)>0,\forall t>0$.
	On the other hand, according to \dref{eq16}, we have 
	\begin{equation}\label{eq17}
		\dot \chi \leq -(\beta +\alpha)\left(\epsilon^T\epsilon-v^T\begin{bmatrix}
	\Omega_1 & \\ & \Omega_2
\end{bmatrix}v\right).
	\end{equation}
	Integrating \dref{eq17} from zero to infinity, it then follows that 
	$$
	\begin{aligned}
		&\chi(\infty)\leq\chi(0)\\&-\int_{0}^{\infty}(\beta +\alpha)\left(\epsilon^T\epsilon-v^T\begin{bmatrix}
	\Omega_1 & \\ & \Omega_2
\end{bmatrix}v\right)dt.
	\end{aligned}
	$$
	Noting that $\chi(\infty)>0$, we have 
	$$\int_{0}^{\infty}\left(\epsilon^T\epsilon-v^T\begin{bmatrix}
	\Omega_1 & \\ & \Omega_2
\end{bmatrix}v\right)dt\leq\frac{\chi(0)}{\beta+\alpha}.
	$$
According to the Passaval Identity \cite{Desoer1975feedback}, 
\begin{equation}\begin{aligned}
	&\int_{-\infty}^{\infty}\epsilon^*(j\omega)\epsilon(j\omega)d\omega\\
&\leq \int_{-\infty}^{\infty}v(j\omega)^*\begin{bmatrix}
	\Omega_1 & \\ & \Omega_2
\end{bmatrix}v(j\omega)d\omega+\frac{2\pi\chi(0)}{\alpha+\beta}
\end{aligned}
	\end{equation}
Combining with \dref{eq4}, we can get that 
\begin{equation}\label{eq19}
	\begin{aligned}
		&\quad\int_{-\infty}^{\infty}\begin{bmatrix}
			\epsilon(j\omega)\\ d(j\omega)
		\end{bmatrix}^*\begin{bmatrix}
			\epsilon(j\omega)\\ d(j\omega)
		\end{bmatrix} d\omega
		\\
		&\leq  \int_{-\infty}^{\infty}v^*(j\omega)\begin{bmatrix}
			\Omega_1 &\\
			& \Omega_2+\eta^2I
		\end{bmatrix}v(j\omega)d\omega+\frac{2\pi\chi(0)}{\alpha+\beta} 
		\\&\leq \gamma^2\int_{-\infty}^{\infty}v^*(j\omega)v(j\omega)d\omega+\frac{2\pi\chi(0)}{\alpha+\beta}.
		\end{aligned}
\end{equation}
Based on the definition of finite-gain $\mathcal L_2$ stablility and the $\mathcal H_{\infty}$ norm theory, it is not difficult to derive that $\Gamma(\cdot)$ is finite-gain $\mathcal L_2$ stable and $\|\Gamma\|_{\infty}\leq\gamma$. $\hfill\square$
\end{proof}
\begin{remark}
	This theorem tells that for a linear system perturbed by an additive dynamic uncertainty, the effect of the aforementioned dynamic event-triggered controller is the same as that of the static event triggering protocol. Therefore, the result of Theorem \ref{thm2} can be directly applied to the dynamic event-triggered case. Moreover, the dynamic event-triggered robust stabilization problem can still be solved using Algorithm \ref{alg1} with triggering mechanism and triggering function  replaced correspondingly. It is also very similar to exclude Zeno behavior.
Although here we only consider the linear system subject to an  additive dynamic uncertainty, it is  not difficult to see that for a linear system perturbed by other types of uncertainties, similar conclusions still hold and similar algorithms can still work.
\end{remark}

\section{General  IQC-Based  Event-Triggered Mechanisms  }\label{IQCsection}
From the above discussion, we  know that no matter the triggering condition is static or dynamic, as long as it gives a constraint that preserves the operator $\Gamma(\cdot)$ (the operator which maps $v(s)$ to $\begin{bmatrix}
	\epsilon(s)\\ d(s)
\end{bmatrix}$) a finite $\mathcal L_2$ gain, the event-triggered robust stabilization problem can then be transformed into the classical problem in the robust control theory. On the other hand, as pointed out in \cite{rantzer1997}, the small gain condition is a specialization of the integral quadratic constraints (IQC). Thus it is natural to ask whether the above `finite $\mathcal L_2$ gain preserving' triggering  condition  can be generalized to some triggering condition that characterizes quadratic constraint between $\begin{bmatrix}
	\epsilon(s)\\ d(s)
\end{bmatrix}$ and $v(s)$. 

To be more specific, we can rewrite \dref{eq19} as 
\begin{equation}\label{IQC}
\begin{aligned}
	&\int_{-\infty}^{\infty}\begin{bmatrix}
		v(j\omega) \\ w(j\omega)
	\end{bmatrix}^*\begin{bmatrix}
		\Omega &\\
		  & -I 
	\end{bmatrix} \begin{bmatrix}
		v(j\omega) \\ w(j\omega)
	\end{bmatrix}d\omega
	 +\frac{2\pi\chi(0)}{\alpha+\beta}\geq0,
\end{aligned}
\end{equation} where $w(j\omega)=\begin{bmatrix}
	\epsilon(j\omega)\\ d(j\omega)
\end{bmatrix}$ and $\Omega=\begin{bmatrix}
	\Omega_1  &\\
		 & \Omega_2+\eta^2I
\end{bmatrix}$.
 Inspired by \cite{rantzer1997}, one can naturally  ask: What if the Hermitian matrix $\begin{bmatrix}
		\Omega \\
		 & -I 
	\end{bmatrix}$ is replaced by a more general one $\Pi=\begin{bmatrix}
		\Pi_1 & \Pi_2\\
		\Pi_2^* & \Pi_3
	\end{bmatrix}$ or even a dynamical one $\Pi(\omega)=\begin{bmatrix}
		\Pi_1(\omega) & \Pi_2(\omega)\\
		\Pi_2^*(\omega) & \Pi_3(\omega)
	\end{bmatrix}$? Due to the existence of the positive defect term $\frac{2\pi\chi(0)}{\alpha+\beta}$, the answer is:  Only the `small-gain' like IQC can be used here. That is, $\Pi(j\omega)=\begin{bmatrix}
		\Pi_1(\omega) & \Pi_2(\omega)\\
		\Pi_2^*(\omega) & \Pi_3(\omega)
	\end{bmatrix}$ with $\Pi_1(\omega)>0$ and $\Pi_3(\omega)<0$.  
This result is concluded as  the following lemma:

\begin{lemma}\label{l3}(Internal stability criterion based on IQC with defect)
	Assume that $G(s)\in \mathcal{RH}_{\infty}$ and  signals $w$ and $v$ are in the  $\mathcal L_2$ space. Suppose that $v(s)=G(s)w(s)$ and $w(s)=\Delta(v(s))$, where $\Delta(\cdot)$ is a bounded causal operator, and  they satisfy the following IQC condition with defect $\xi>0$ defined by $\Pi$:
	\begin{equation}\label{eq20}
		\begin{aligned}
			\int_{-\infty}^{\infty}\begin{bmatrix}
				 v(j\omega) \\  w(j\omega)
			\end{bmatrix}^*\begin{bmatrix}
				\Pi_{1}(\omega) & \Pi_2(\omega)\\
				\Pi_2^*(\omega) & \Pi_3(\omega)
			\end{bmatrix}
			\begin{bmatrix}
				 v(j\omega) \\  w(j\omega)
			\end{bmatrix}d\omega+\xi\geq 0		\end{aligned}
	\end{equation}
	with $\Pi_1(\omega)>0$ and $\Pi_3(\omega)<0$,  $\forall \omega\in \mathbf R$, two Hermitian valued matrices with compatible dimensions.
	Then, if  
	\begin{equation}\label{eq21}
		\begin{bmatrix}
			G(j\omega)\\ I
		\end{bmatrix}^*\begin{bmatrix}
				\Pi_{1}(\omega) & \Pi_2(\omega)\\
				\Pi_2^*(\omega) & \Pi_3(\omega)
			\end{bmatrix}\begin{bmatrix}
			G(j\omega)\\ I
		\end{bmatrix}<0,\quad  \forall
    \omega\in \mathbf{R},
	\end{equation}
	then the system interconnection of $G$ and $\Delta$ is internally stable.
\end{lemma}
\begin{proof}
	Since $\Pi_1(\omega)>0$ and $\Pi_3(\omega)<0$ $\forall \omega\in \mathbf R$, it is easy to verify that $\begin{bmatrix}
				\Pi_{1}(\omega) & \Pi_2(\omega)\\
				\Pi_2^*(\omega) & \Pi_3(\omega)
			\end{bmatrix}$ is congruent to $\begin{bmatrix}
				I & 0 \\ 0 & -I
			\end{bmatrix}$, i.e., there exists a nonsingular matrix $C(j\omega) $ such that 
			$$\begin{bmatrix}
				\Pi_{1}(\omega) & \Pi_2(\omega)\\
				\Pi_2^*(\omega) & \Pi_3(\omega)
			\end{bmatrix}=C^*(j\omega)\begin{bmatrix}
				I & 0 \\ 0 & -I
			\end{bmatrix}C(j\omega) \quad\forall \omega\in \mathbf R.
			$$
   Let $\begin{bmatrix}
   	\tilde v(j\omega)\\ \tilde w(j\omega)
   \end{bmatrix}=C(j\omega)\begin{bmatrix}
   	v(j\omega) \\ w(j\omega)
   \end{bmatrix}$ and denote $\tilde w=\tilde\Delta(\tilde v )$ , it then follows from \dref{eq20} that \begin{equation}\label{eq22}
   	\|\tilde w\|_2^2\leq\|\tilde v\|_2^2+\frac{\xi}{2\pi}.  
   \end{equation} Therefore  $\|\tilde \Delta\|_{\infty}\leq 1$. In light of  Small Gain Theorem (Lemma \ref{smallgain}), the closed-loop system is internally stable if \begin{equation}\label{eq23}
   \begin{bmatrix}
   	\tilde G(j\omega) \\ I 
   \end{bmatrix}^*\begin{bmatrix}
   	I & 0 \\ 0 & -I
   \end{bmatrix}\begin{bmatrix}
   	\tilde G(j\omega) \\ I 
   \end{bmatrix}<0 \quad \forall \omega \in \mathbf R,\end{equation}
   where $\tilde G(j\omega)$ denotes the transfer matrix from $\tilde w$ to $\tilde v$.
   Noticing that $\begin{bmatrix}
   	0 & \tilde G(j\omega)  \\ 0 &I 
   \end{bmatrix}\begin{bmatrix}
   	\tilde v \\ \tilde w
   \end{bmatrix}=C(j\omega)\begin{bmatrix}
   	0 &  G(j\omega)  \\ 0 &I 
   \end{bmatrix}\begin{bmatrix}
   	v\\w 
   \end{bmatrix}$ and $\begin{bmatrix}
   	v\\ w
   \end{bmatrix}=C^{-1}(j\omega)\begin{bmatrix}
   	\tilde v \\ \tilde w
   \end{bmatrix}$, we have 
   $$\begin{aligned}
   	\begin{bmatrix}
   	\tilde G(j\omega)\\I 
   \end{bmatrix}& =C(j\omega)\begin{bmatrix}
   	0 &  G(j\omega)  \\ 0 &I 
   \end{bmatrix}C^{-1}(j\omega)\begin{bmatrix}
   	0\\ I
   \end{bmatrix}\\
   & = C(j\omega)\begin{bmatrix}
   	 G(j\omega)\\I 
   \end{bmatrix}\begin{bmatrix}
   	0 & I
   \end{bmatrix}C^{-1}(j\omega)\begin{bmatrix}
   	0\\ I
   \end{bmatrix}\\
   & = C(j\omega)\begin{bmatrix}
   	 G(j\omega)\\I 
   \end{bmatrix} D(j\omega)
   \end{aligned}$$ and $D(j\omega)$ is nonsingular $\forall \omega\in \mathbf R$.
 Therefore, \dref{eq23} is equivalent to 
 \dref{eq21}. $\hfill\square$
\end{proof} 
\begin{remark}
	It is not difficult  to see that we can relax the condition $\Pi_1(\omega)>0$ to $\Pi_1(\omega)\geq 0$ while not change the result derived above. 
\end{remark}
\begin{remark}
	This lemma extends the results in the classical IQC-based internal stability theorem in \cite{rantzer1997} in the sense that here we allow a defect term $\xi$ in the quadratic functional. However, a stronger constraint of the positivity  has to be imposed on the matrices $\Pi_1(\omega)$ and $-\Pi_3(\omega)$. In short, only when the IQC is  `small-gain' like  can we    introduce a defect term in the quadratic functional  while remain the internal stability condition unchanged. The defect term corresponds to the additional term $\chi$ in the event-triggered condition and is critical  and indispensable since we have to exclude Zeno behavior.   
\end{remark}
Now, we are ready to put forward  the IQC-based dynamic event-triggered control algorithm that is robust against additive dynamic uncertainty described in \dref{eq2}.
\begin{alg}\label{alg3} An algorithm to find a robust IQC based dynamic event-triggered protocol for $P_{\Delta}=P_0+\Delta$:

	Step 1: Solve the standard optimal 
	$H_{\infty}$ synthesis problem  in Theorem \ref{thm2} and find the optimal $H_{\infty} $ level $\gamma_{opt}$. If $\gamma_{opt}\eta<1$, go to the next  step . Otherwise the robust stabilizing event-triggered controller may not exist. 
	
	Step 2: Select a $\mathbf C^{(p+q)}\mapsto \mathbf C^{(p+q)}$ square transfer matrix $G_1(j\omega)=\left[\begin{smallmatrix} 
\begin{array}{c|c}
	A_{1}  & B_{1} \\ \hline 
	C_{1} &  D_{1}
\end{array}
\end{smallmatrix}\right]\in \mathcal{RH}_{\infty}$ such that $\|G_1(j\omega)\|_{\infty}^2=\sigma_1^2<\gamma_{opt}^{-2}-\eta^2$ and let $\bar v(t)\in \mathbf R^{p+q}:=\left[\begin{smallmatrix} 
\begin{array}{c|c}
	A_{1}  & B_{1} \\ \hline 
	C_{1} &  D_{1}
\end{array}
\end{smallmatrix}\right]v(t)$ (Here $v(t)$ is defined as in  Theorem \ref{thm1}).

Step 3: Select a nonsingular square transfer matrix $G_2(j\omega)=\left[\begin{smallmatrix} 
\begin{array}{c|c}
	A_{2}  & B_{2} \\ \hline 
	C_{2} &  D_{2}
\end{array}
\end{smallmatrix}\right]\in \mathcal{RH}_{\infty}$ such that $D_2$ is nonsingular and  $\|G_2(j\omega)\|_{\infty}\leq 1$. Denote $\bar\epsilon(t)\in\mathbf R^{p+q} :=\left[\begin{smallmatrix} 
\begin{array}{c|c}
	A_{2}-B_2D_2^{-1}C_2  & B_{2}D_2^{-1} \\ \hline 
	D_2^{-1}C_{2} &  D_{2}^{-1}
\end{array}
\end{smallmatrix}\right]\begin{bmatrix}
	\epsilon_y\\ \epsilon_u
\end{bmatrix}$, that is $\bar\epsilon(s)=G_2(s)^{-1}\epsilon(s)$.
	
	Step 4: Set the triggering function to be $$f=\bar\epsilon(t)^T\bar\epsilon(t)-\bar v(t)^T\bar v(t)-\chi(t)$$
	with $\chi(0)>0$ and 
	$$ \dot\chi(t)=-\beta\chi(t)-\alpha(\bar\epsilon(t)^T\bar\epsilon(t)-\bar v(t)^T\bar v(t)),$$
	where $\beta>0$ and $\alpha>0$.
	When $f\geq 0$ at time $t_k$, the sampler $S_1$ and $S_2$ samples the output of the plant ($y(t_k)$) and the output of the controller ($u(t_k)$) respectively, i.e., $\hat y$ is set to be $y(t_k)$, $\hat u$ is set to be $u(t_k)$, and $\epsilon_y$ and $\epsilon_u$ are reset to be $0$. Otherwise, $\hat y$ and $\hat u$ remain unchanged while $\epsilon_y=\hat y-y(t)$ and $\epsilon_u=\hat u-u(t)$.
	
	Step 5: Choose $\gamma_{opt}<\gamma<\frac{1}{\sqrt{\sigma_1^2+\eta^2}}$ and solve the $\gamma$-suboptimal $H_{\infty}$ problem in Theorem \ref{thm2} and get the controller $K(s)=\left[\begin{smallmatrix} 
\begin{array}{c|c}
	A_{k}  & B_{k} \\ \hline 
	C_{k} &  D_{k}
\end{array}
\end{smallmatrix}\right]$.

\end{alg}

The next Theorem shows the effectivity of the aforementioned event-triggered controller design algorithm.

\begin{thm}
	For the linear system \dref{eq1} subject to the additive dynamic uncertainty \dref{eq2} with $H_{\infty}$ norm bound $\eta$, if we apply the event-triggering mechanism and the controller $K(s)$ designed in Algorithm \ref{alg3}, then the closed-loop system is internally stable. Moreover, no Zeno behavior exhibits.
\end{thm}
\begin{proof}
	Following the proof of Theorem \ref{thm7}, it is not difficult to derive that 
	$$\int_0^{+\infty}(\bar \epsilon^T(t)\bar\epsilon(t)-\bar v^T(t)\bar v(t))dt\leq \dfrac{\chi(0)}{\alpha+\beta}.$$
	Therefore, we have $$\int_{-\infty}^{\infty}\begin{bmatrix}
		\bar v(j\omega)\\ \bar \epsilon(j\omega)
	\end{bmatrix}^*\begin{bmatrix}
		I & 0\\ 0 & -I
	\end{bmatrix}\begin{bmatrix}
		\bar v(j\omega)\\ \bar \epsilon(j\omega)
	\end{bmatrix}d\omega+\frac{2\pi\chi(0)}{\alpha+\beta}\geq0.$$
	Notice that $\begin{bmatrix}
		\bar v(j\omega)\\ \bar \epsilon(j\omega)
	\end{bmatrix}=\begin{bmatrix}
		G_1(j\omega) & 0\\ 0 & G_2^{-1}(j\omega)
	\end{bmatrix}\begin{bmatrix}
		 v(j\omega)\\  \epsilon(j\omega)
	\end{bmatrix}$, it then follows that\begin{equation}\label{eq24}
		\int_{-\infty}^{\infty}\begin{bmatrix}
		 v(j\omega)\\  \epsilon(j\omega)
	\end{bmatrix}^*\begin{bmatrix}
		\Psi_1(\omega) & 0\\ 0 & -\Psi_2(\omega)
	\end{bmatrix}\begin{bmatrix}
		 v(j\omega)\\  \epsilon(j\omega)
	\end{bmatrix}d\omega+\frac{2\pi\chi(0)}{\alpha+\beta}\geq0
	\end{equation}
	with $\Psi_1(\omega)=G_1^*(j\omega)G_1(j\omega)\geq0 $ and $\Psi_2(\omega)=G_2^{-*}(j\omega)G_2^{-1}(j\omega)>0$,  $\forall \omega\in \mathbf R$.
	On the other hand, in light of \dref{eq4}, $$\int_{-\infty}^{\infty}(v^*(j\omega)\eta^2Iv(j\omega)-d^*(j\omega)d(j\omega))d\omega\geq 0.$$
	Combining with \dref{eq24}, we then have 
	\begin{equation}\label{eq25}
	\begin{aligned}
		&\int_{-\infty}^{\infty}\begin{bmatrix}
		   v(j\omega)\\ \epsilon(j\omega)\\ d(j\omega)
		\end{bmatrix}^*\begin{bmatrix}
			\Psi_1(\omega)+\eta^2I & & \\
			 & -\Psi_2(\omega) & \\
			 & & -I
		\end{bmatrix}\begin{bmatrix}
		   v(j\omega)\\ \epsilon(j\omega)\\ d(j\omega)
		\end{bmatrix}d\omega\\
		&+\frac{2\pi\chi(0)}{\alpha+\beta}\geq0.
	\end{aligned}
	\end{equation}
	Applying Lemma \ref{l3}, the system is  internally stable, if 
	\begin{equation}\label{eq26}
		\begin{aligned}
			\begin{bmatrix}
				M(j\omega)\\ I
			\end{bmatrix}^*\begin{bmatrix}
			\Psi_1(\omega)+\eta^2I & & \\
			 & -\Psi_2(\omega) & \\
			 & & -I
		\end{bmatrix}\begin{bmatrix}
				M(j\omega)\\ I
			\end{bmatrix}<0
		\end{aligned}
	\end{equation}
	$\forall \omega\in \mathbf R.$
	This condition holds if and only if 
	$$M^*(j\omega)(\Psi_1(\omega)+\eta^2I)M(j\omega)<\begin{bmatrix}
		\Psi_2(\omega) & \\ & I
	\end{bmatrix},$$
	which is equivalent to 
	\begin{equation}\label{eq27}
		\begin{aligned}
		&\begin{bmatrix}
			G_2^*(j\omega) & \\ & I
		\end{bmatrix}M^*(j\omega)(G_1^*(j\omega)G_1(j\omega)+\eta^2I)M(j\omega)\\&\times\begin{bmatrix}
			G_2(j\omega) & \\ & I
		\end{bmatrix}<I.
	\end{aligned}
	\end{equation}
	 When the event triggering mechanism and the controller $K(s)$ is designed as in the Algorithm \ref{alg3}, we have $\|G_2(j\omega)\|_{\infty}\leq1$ and meanwhile $$\sqrt{\|G_1(j\omega)\|^2_{\infty}+\eta^2}\|M(j\omega)\|_{\infty}<1.$$ It is easy to verify that \dref{eq27} holds. Thus the closed-loop system is internally stable. The excluding of the Zeno behavior is similar to the proof of Theorem \ref{thm3} and thus is omitted for brevity.  $\hfill\square$
 
\end{proof}

\begin{remark}
	Actually, as shown in \dref{eq24},  the event-triggered protocol designed in Algorithm \ref{alg3} establishes a more general IQC compared to the one shown in \dref{IQC}. The to-be-developed  dynamical systems $G_1(s)$ and $G_2(s)$ endow a great amount of flexibility and degree of freedom to the controller design and the performance optimization.
\end{remark}

\section{Simulation Results }\label{s7}
In this section, a design example is illustrated for a linear system with additive uncertainty to manifest the effectiveness of the proposed algorithm. 
Consider the following  strictly proper linear system: $$P(s)=\left[\begin{smallmatrix} 
\begin{array}{cc|c}
	-12.5 & 5.9 & 1 \\ -7.1 & 13.8 & 2 \\ \hline 
	-4  & 5.5  &  0
\end{array}
\end{smallmatrix}\right]$$ and assume that the additive perturbation is of the form: 
$$\Delta(s)=\left[\begin{smallmatrix} 
\begin{array}{cc|c}
	-15.4 & 10.7 & -1.24 \\ -15.7 & -1.41 & -1.28 \\ \hline 
	1.165  & -2.07  &  0
\end{array}
\end{smallmatrix}\right].$$ 
It is easy to calculate that $\|\Delta\|_{\infty}=\eta=0.1112$. 

\subsection{Verifying the effectiveness of Algorithm \ref{alg1}}
Based on Algorithm \ref{alg1}, by solving the standard optimal $H_{\infty}$ synthesis problem in Theorem \ref{thm2}, we find that the optimal $H_{\infty}$ level $\lambda_{opt}=3.0683<\eta^{-1}$. Therefore the robust stabilization problem is solvable. Following Step $1$, we set $\mu=0.1$, $\nu=5$, $\Omega_1=\lambda_{opt}^{-2}\times0.98=0.1041$ and $\Omega_2=(\lambda_{opt}^{-2}-\eta^{2})\times 0.98=0.0920$. It is then easy to derive that $\gamma=0.3230$. Solving the $\gamma^{-1}$-suboptimal $H_{\infty}$ synthesis problem, we can then get a controller $$K(s)=\left[\begin{smallmatrix} 
\begin{array}{c|c}
	-11.9130 & -0.7953  \\ \hline 
	-3.5437  & -1.8130  
\end{array}
\end{smallmatrix}\right].$$
Applying  the  controller and event-triggered mechanism designed as above, we  then depict the evolution of the internal states $x_p$ and $ x_k$ with time. As shown in  Fig. \ref{fig4}, the closed-loop system is internally stabilized by the designed event-triggered control law. Moreover, we depict the evolution of $\|\epsilon\|$  and $\sqrt{\|v\|^2+\mu e^{-\nu t}}$ between $0s$ to $1.5s$ in Fig. \ref{fig5} whose intersections represent the triggering instants. This figure shows clearly that there is no Zeno behavior. 
\begin{figure}
\begin{center}
\includegraphics[height=6cm]{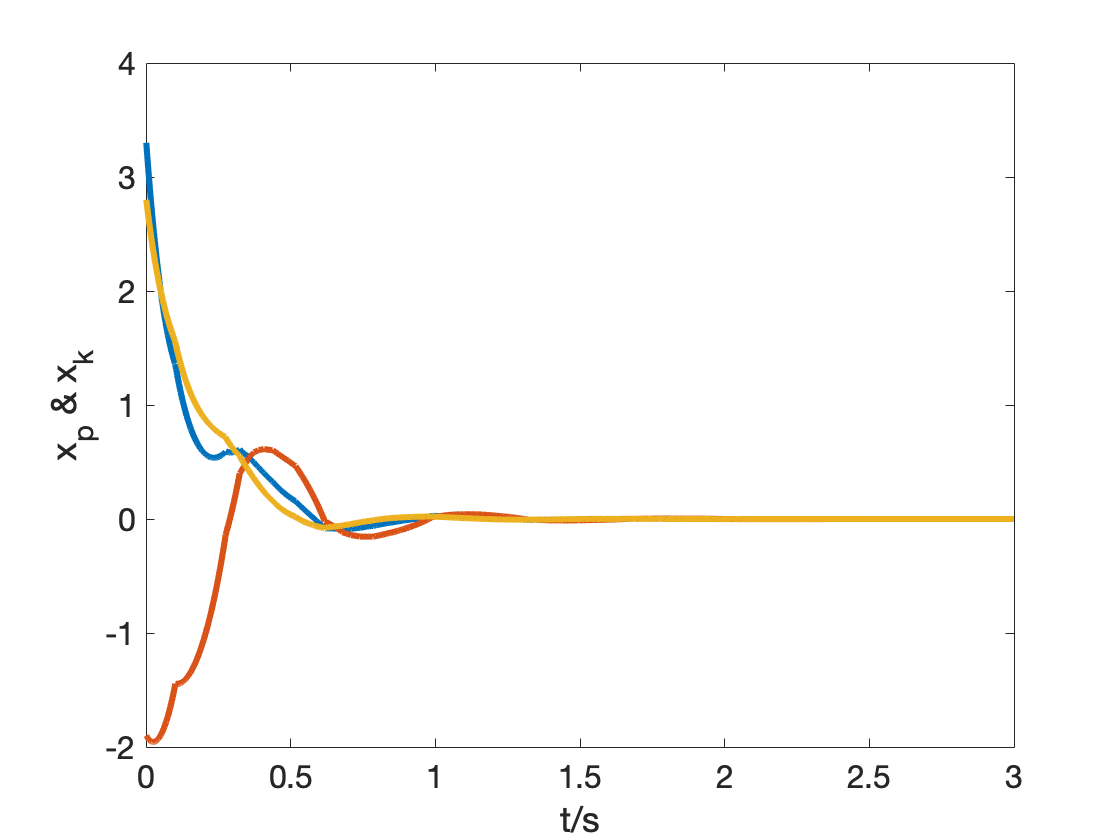}      
\caption{The evolution of the internal states $x_p$ and $x_k$. }  
\label{fig4}                                 
\end{center}                                 
\end{figure}

\begin{figure}
\begin{center}
\includegraphics[height=6cm]{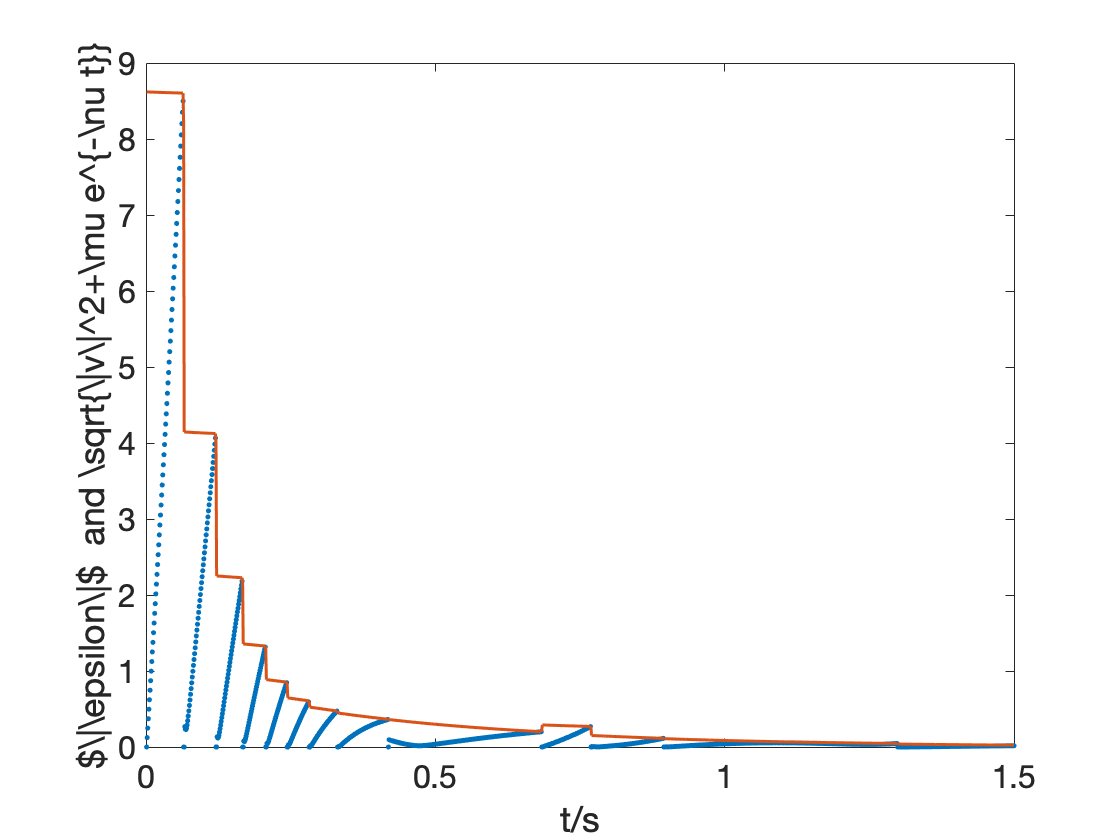}      
\caption{The evolution of  $\|\epsilon\|$  and $\sqrt{\|v\|^2+\mu e^{-\nu t}}$, where the blue dotted line denotes the evolution of $\|\epsilon\|$ and the red line denotes the evolution of $\sqrt{\|v\|^2+\mu e^{-\nu t}}$. The triggering instants are those when the blue dotted line crosses the red line. }  
\label{fig5}                                 
\end{center}                                 
\end{figure}

\subsection{Verifying the effectiveness of  Algorithm \ref{alg3}}
Next, we move on to test  the effectivity of the Algorithm \ref{alg3}.
We still consider the plant $P(s)$ subject to additive dynamic uncertainty $\Delta(s)$ in the previous case. Based on  Step $1$, the robust event-triggered stabilization problem is solvable. In  Step $2$ and Step $3$, we set $$G_1=\left[\begin{smallmatrix} 
\begin{array}{cc|cc}
	-1.9 & 6.7 & 0.1019 & -0.2209 \\  
	-4.3  & -10.4 & 0.7561 & -0.4842 \\ \hline
	0.3 & -4.1 & 0.0025 &  0 \\
	0.58 & 0.39 & 0.0671 & 0.1529 
\end{array}
\end{smallmatrix}\right]$$ and 
$$G_2=\left[\begin{smallmatrix} 
\begin{array}{cc|cc}
	-19 & 16.7 & 0.2 & -2.6 \\  
	-43  & 0.4 & 1.9 & -0.57 \\ \hline
	0.064 & 0.8743 & 0.7677 &  0 \\
	0.1237 & -0.0832 & 0.0671 & -0.5118 
\end{array}
\end{smallmatrix}\right].$$ Note that $\|G_1\|_{\infty}=\sqrt{0.95(\gamma_{opt}^{-2}-\eta^2)}=0.2986$ and $\|G_2\|_{\infty}=0.9<1$. The initial value of the internal state of $G_1$ and $G_2^{-1}$ is randomly chosen. 
In Step $4$, we set $\alpha=2.5$, $\beta=1.0$ and randomly choose a positive $\chi(0)$. In Step $5$, we set $\gamma= 0.8(\frac{1}{\sqrt{\|G_1\|_{\infty}^2+\eta^2}}-\gamma_{opt}) + \gamma_{opt}=3.1244$. Solving the $\gamma$-suboptimal $H_{\infty}$ problem in Theorem \ref{thm2}, we can get that $$K(s)=\left[\begin{smallmatrix} 
\begin{array}{cc|c}
	-9.0230 & 126.3555 & 0.7510  \\  
	-324.2813  & -14216 & -174.0507  \\ \hline
	-0.1597 & 148.4024 & 0 \\ 
\end{array}
\end{smallmatrix}\right].$$

Applying the controller and the event-triggered mechanism developed here, we can then draw the evolution of the internal states $x_p$ and $x_k$ from $0s$ to $6s$ in Fig. \ref{fig6}. It is clear from Fig. \ref{fig6} that the closed-loop system is  internally stable.
\begin{figure}
\begin{center}
\includegraphics[height=6cm]{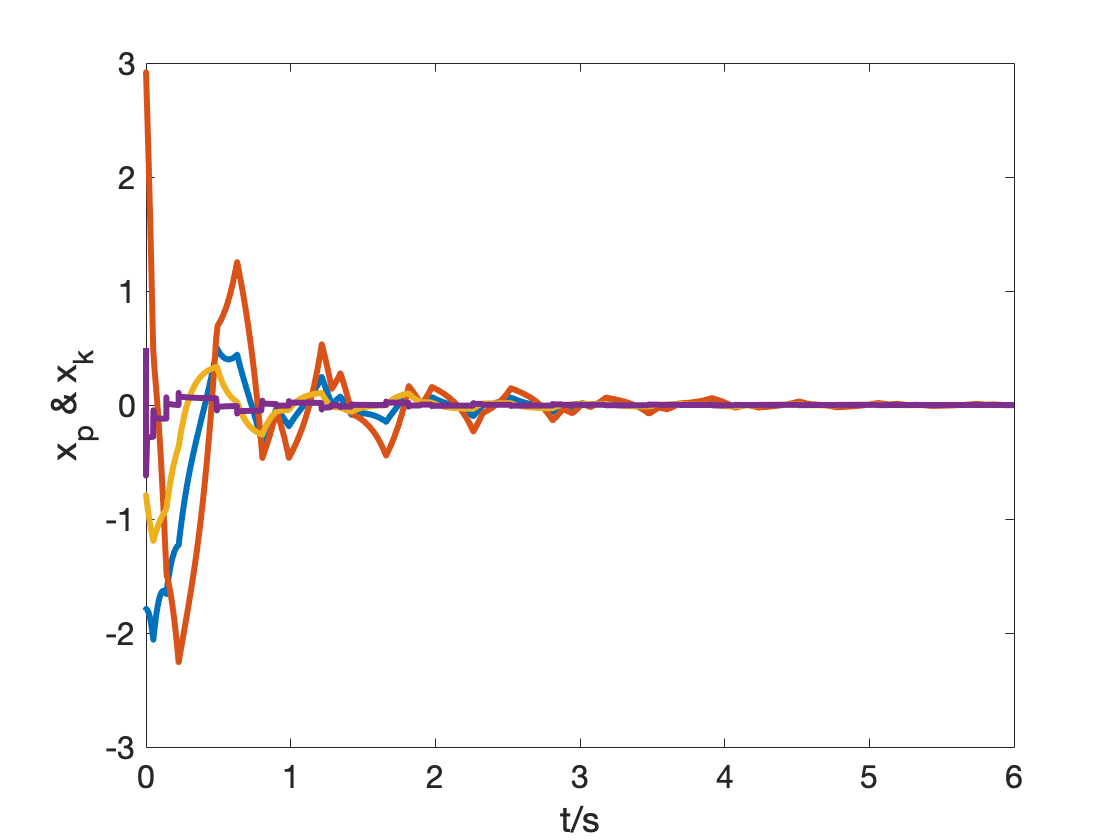}      
\caption{The evolution of the internal states $x_p$ and $x_k$ from  $0s$ to $6s$. }  
\label{fig6}                                 
\end{center}                                 
\end{figure}

To show that the closed-loop system does not exhibit Zeno behavior, we draw the evolution of $\|\epsilon\|$ and $\sqrt{\|v\|^2+\chi}$ respectively from  $0s$ to $3s$ in Fig. \ref{fig7}. The triggering instants are their intersections and it is clear that there is no Zeno behavior. 

\begin{figure}
\begin{center}
\includegraphics[height=6cm]{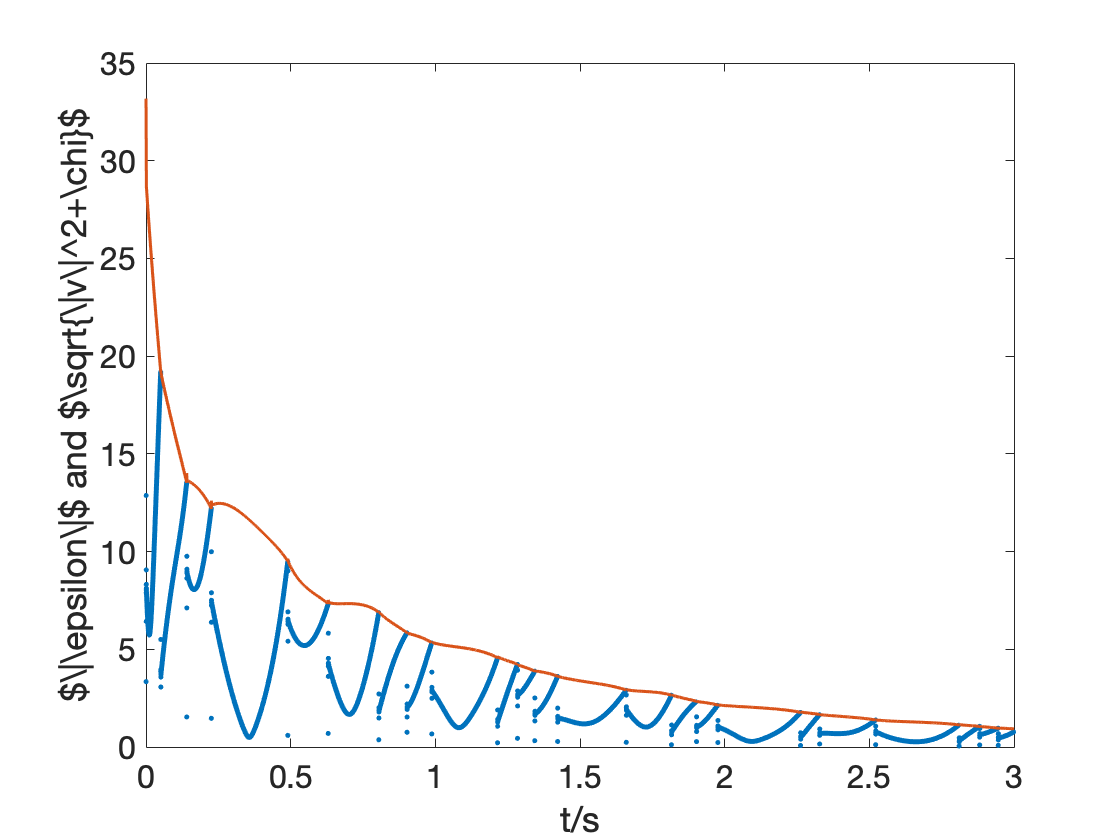}      
\caption{The evolution of $\|\epsilon\|$ and $\sqrt{\|v\|^2+\chi}$ from  $0s$ to $3s$, where the blue dotted line represents that of $\|\epsilon\|$ and the red line represents that of $\sqrt{\|v\|^2+\chi}$. The triggering instants are those when the blue dotted line crosses the red line. }  
\label{fig7}                                 
\end{center}                                 
\end{figure}

\section{Conclusion}\label{s8}
In this paper, we have established an operator-theoretic approach for the robust event-triggered control problem of general linear systems subject to frequency-domain uncertainties. 
By showing that under the typical static and dynamic, and even the generalized IQC-based event triggering mechanisms, the mapping from sampled outputs to sampling errors are essentially finite-gain $\mathcal L_2$ stable affine operators, robust event-triggered control laws can be systematically designed by solving the standard $H_{\infty}$ synthesis problem of a modified linear system. 

There are many potential extensions to the present work, e.g., considering the robust $H_2$ and $H_{\infty}$ performance of event-triggered controllers and designing robust consensus event-triggered protocols for linear or nonlinear multi-agent systems.       
\bibliography{autosam}           


\end{document}